\documentclass[10pt, letter, a4paper]{./lib/IEEEtran}
\usepackage{amsmath}
\usepackage{mathtools}
\usepackage{array}
\usepackage{cite}
\usepackage{algpseudocode,algorithm,algpseudocode}
\usepackage{float}
\usepackage{tabularx}
\usepackage{mathrsfs} 
\usepackage{tikz}
\usetikzlibrary{chains, arrows,shapes,positioning,arrows,decorations.markings,calc}
\usepackage{pgfplots}
\pgfplotsset{compat=newest}
\usepgfplotslibrary{groupplots}
\usepgfplotslibrary{fillbetween}
\usepackage{gensymb}
\usepackage{soul}	
\usepackage{xcolor}
\usepackage[shortcuts,acronym]{glossaries}
\makeglossaries
\usepackage{fixmath}
\usepackage{amssymb}
\usepackage{bm}
\usepackage{bbm}
\usepackage{pbox}
\usepackage{multirow}
\usepackage{tabularx}
\pgfplotsset{compat=newest,
	/pgfplots/ybar legend/.style={
		/pgfplots/legend image code/.code={%
			\draw[##1,/tikz/.cd,bar width=3pt,yshift=-0.2em,bar shift=0pt]
			plot coordinates {(0cm,0.8em)};},
	},}

\newcommand\myeqa{\stackrel{\mathclap{\mbox{($a$)}}}{=}}
\newcommand\myeqb{\stackrel{\mathclap{\mbox{($b$)}}}{=}}
\newcommand\myeqc{\stackrel{\mathclap{\mbox{(\scriptsize$\epsilon = 1$)}}}{=}}

\usepackage{amsthm}
\theoremstyle{plain}

\newtheorem{lemma}{Lemma}

\usepackage{caption}

\makeatletter
\def\therule{\makebox[\algorithmicindent][l]{\hspace*{.5em}\vrule height .75\baselineskip depth .25\baselineskip}}%

\newtoks\therules
\therules={}
\def\appendto#1#2{\expandafter#1\expandafter{\the#1#2}}
\def\gobblefirst#1{
	#1\expandafter\expandafter\expandafter{\expandafter\@gobble\the#1}}%
\def\LState{\State\unskip\the\therules}
\def\pushindent{\appendto\therules\therule}%
\def\popindent{\gobblefirst\therules}%
\def\printindent{\unskip\the\therules}%
\def\printandpush{\printindent\pushindent}%
\def\popandprint{\popindent\printindent}%

\algdef{SE}[WHILE]{While}{EndWhile}[1]
{\printandpush\algorithmicwhile\ #1\ \algorithmicdo}
{\popandprint\algorithmicend\ \algorithmicwhile}%
\algdef{SE}[FOR]{For}{EndFor}[1]
{\printandpush\algorithmicfor\ #1\ \algorithmicdo}
{\popandprint\algorithmicend\ \algorithmicfor}%
\algdef{S}[FOR]{ForAll}[1]
{\printindent\algorithmicforall\ #1\ \algorithmicdo}%
\algdef{SE}[LOOP]{Loop}{EndLoop}
{\printandpush\algorithmicloop}
{\popandprint\algorithmicend\ \algorithmicloop}%
\algdef{SE}[REPEAT]{Repeat}{Until}
{\printandpush\algorithmicrepeat}[1]
{\popandprint\algorithmicuntil\ #1}%
\algdef{SE}[IF]{If}{EndIf}[1]
{\printandpush\algorithmicif\ #1\ \algorithmicthen}
{\popandprint\algorithmicend\ \algorithmicif}%
\algdef{C}[IF]{IF}{ElsIf}[1]
{\popandprint\pushindent\algorithmicelse\ \algorithmicif\ #1\ \algorithmicthen}%
\algdef{Ce}[ELSE]{IF}{Else}{EndIf}
{\popandprint\pushindent\algorithmicelse}%
\algdef{SE}[PROCEDURE]{Procedure}{EndProcedure}[2]
{\printandpush\algorithmicprocedure\ \textproc{#1}\ifthenelse{\equal{#2}{}}{}{(#2)}}%
{\popandprint\algorithmicend\ \algorithmicprocedure}%
\algdef{SE}[FUNCTION]{Function}{EndFunction}[2]
{\printandpush\algorithmicfunction\ \textproc{#1}\ifthenelse{\equal{#2}{}}{}{(#2)}}%
{\popandprint\algorithmicend\ \algorithmicfunction}%
\makeatother

\algrenewcommand\algorithmicprocedure{\textbf{Input}}
\algrenewcommand\algorithmicreturn{\textbf{Output:}}

\begin{document}
\bstctlcite{IEEEexample:BSTcontrol}

\title{A Spatiotemporal Framework for Information Freshness in IoT Uplink Networks}
	
	\author{\IEEEauthorblockN{Mustafa Emara\IEEEauthorrefmark{1}\IEEEauthorrefmark{2}, Hesham ElSawy\IEEEauthorrefmark{3}, Gerhard Bauch\IEEEauthorrefmark{2}} \\
		\IEEEauthorblockA{\IEEEauthorrefmark{1} Next Generation and Standards, Intel Deutschland GmbH, Neubiberg, Germany} \\
		\IEEEauthorblockA{\IEEEauthorrefmark{2} Institute of Communications, Hamburg University of Technology, Hamburg, Germany} \\
		\IEEEauthorblockA{\IEEEauthorrefmark{3} Electrical Engineering Department, King Fahd University of Petroleum and Minerals, Dhahran, Saudi Arabia} \\
		Email:\{mustafa.emara@intel.com,  hesham.elsawy@kfupm.edu.sa, bauch@tuhh.de\}
	}
	\maketitle
	\thispagestyle{empty}
	\maketitle
	\thispagestyle{empty}
	
\newacronym{AoI}{AoI}{age of information}
\newacronym{BS}{BS}{base station}
\newacronym{DTMC}{DTMC}{discrete time Markov chain}
\newacronym{EA}{EA}{equal allocation}
\newacronym{5G}{5G}{fifth generation}
\newacronym{FCFS}{FCFS}{first come first serve}
\newacronym{IoT}{IoT}{Internet of Things}
\newacronym{KPI}{KPI}{key performance indicator}
\newacronym{LT}{LT}{Laplace transform}
\newacronym{LTE}{LTE}{long term evolution}
\newacronym{MAM}{MAM}{matrix analytic method}
\newacronym{MTC}{MTC}{machine type communication}
\newacronym{MAC}{MAC}{medium access control}

\newacronym{NB-IoT}{NB-IoT}{narrowband IoT}
\newacronym{PPP}{PPP}{Poisson point processes}
\newacronym{PDF}{PDF}{probability density function}
\newacronym{PMT}{PMT}{prioritized multi-stream traffic}
\newacronym{QoS}{QoS}{quality of service}
\newacronym{QCI}{QCI}{QoS class identifier}
\newacronym{QBD}{QBD}{quasi-birth-death}
\newacronym{RAT}{RAT}{radio access technology}
\newacronym{SINR}{SINR}{signal to interference noise ratio}
\newacronym{SIR}{SIR}{signal to interference ratio}
\newacronym{3GPP}{3GPP}{Third generation partnership project}
\newacronym{TSN}{TSN}{time senstive networking}
\newacronym{TSP}{TSP}{transmission success probability}
\newacronym{URLLC}{URLLC}{ultra reliable low latency communication}
\newacronym{WA}{WA}{weighted allocation}
	\thispagestyle{empty}
\begin{abstract}
Timely message delivery is a key enabler for Internet of Things (IoT) and cyber-physical systems to support wide range of context-dependent applications. Conventional time-related metrics, such as delay, fails to characterize the timeliness of the system update or to capture the freshness of information from application perspective. Age of information (AoI) is a time-evolving measure of information freshness that has received considerable attention during the past years. In the foreseen large-scale and dense IoT networks, joint temporal (i.e., queue aware) and spatial (i.e., mutual interference aware) characterization of the AoI is required. In this work we provide a spatiotemporal framework that captures the peak AoI for large scale IoT uplink network. To this end, the paper quantifies the peak AoI for large-scale cellular network with Bernoulli uplink traffic. Simulation results are conducted to validate the proposed model and show the effect of traffic load and decoding threshold. Insights are driven to characterize the network stability frontiers and the location-dependent performance within the network.
\end{abstract}
%

\section{Introduction}\label{section:introduction}
The timeliness and retainability of continuous updates of nodes within a system are overarching requirements among different technology segments, such as vehicular, industrial \ac{IoT} and cellular \cite{3GPP2018, Kim2012}. This implies continuous information update about the real-time state between a given VRU and its targeted cluster of vehicles \cite{NGMNA2016}. The proposed \ac{AoI} metric in \cite{Kaul2012} characterizes the freshness of information at the receiver and has received increasing attention in the past years. The age at a given time stamp (i.e., observation point) is defined as the current time stamp minus the time at which the observed state (or packet) was generated \cite{Yates2019}. Compared to traditional time metrics (e.g. delay and jitter), \ac{AoI} assists in achieving timely updates in a way those traditional metrics do not \cite{AoIDhillon, Ceran2018}. 

The traffic generated by the existing \ac{IoT} devices highly impact the \ac{AoI} and the overall network performance. \ac{IoT} traffic can be categorized into time and event triggered \cite{Metzger2019}. Time-triggered events generate periodic traffic as in vehicular communications, smart grids and wireless sensor networks \cite{Palattella2016}. In such segments, a central entity collects status updates from multiple nodes (e.g., sensors, vehicles and monitors) through wireless channels. A critical challenge is how to maintain timely status updates over all the connected nodes \cite{Jiang2019}. To this end, characterizing the \ac{AoI} leads to informed designs to enhance the performance of time-critical applications. Moreover, event triggered traffic are dependent on external events or critical information reporting. In this paper, we consider Bernoulli-based traffic model that mimic the aforementioned scenarios to analyze the peak \ac{AoI} in uplink \ac{IoT} networks. 
 
Based on the foreseeable large number of deployed devices \cite{Ayoub2018}, interference might hinder timely updates of a given link of interest. In order to characterize the network-wide mutual interference, stochastic geometry is a prominent framework that can be leveraged \cite{Andrews2011, Elsawy_tutorial, Haenggi2012}. However, considerations of the temporal evolution was not considered, as the full-buffer assumption at the transmitter side is widely adopted. To account for the temporal domain, recent efforts have integrated queueing theory with stochastic geometry, which offers a full spatiotemporal characterization for large-scale networks \cite{Zhong2017, Gharbieh2018, Yang2019, Chisci2019, YangAoI2019}. Following such a spatiotemporal analysis of the network, \ac{AoI} was studied under a stochastic geometry framework in \cite{Hu2018}, where lower and upper bounds for the average \ac{AoI} were derived. Moreover, \ac{AoI} under a spatiotemporal framework has recently been investigated in \cite{YangAoI2019}, where the authors investigated different scheduling techniques to optimize the peak \ac{AoI} under a spatiotemporal framework. To this end, large scale spatiotemporal IoT analytical framework for peak \ac{AoI} characterization under different traffic loads is still an open research problem.

Throughout this work, we provide an analytical framework to characterize peak \ac{AoI} that entails macroscopic and microscopic scales of uplink large scale \ac{IoT} networks. Peak \ac{AoI} is considered throughout this work due to its peak-characterization of the \ac{AoI} compared to the average \ac{AoI}, which fails to capture the underlying variance \cite{Huang2015, YangAoI2019}. For the macroscopic aspect, tools from stochastic geometry are adopted to account for the mutual interference among active devices (i.e., position dependent). Tools from queueing theory are employed to characterize the microscopic queue evolution at each device. To track the queue status at a given time stamp, a \ac{DTMC} is employed for each device. Expressions for the distribution of the coverage probability (i.e., meta distribution) are derived, which entails the effect of the traffic arrival. In addition, temporal distribution and the peak \ac{AoI} are presented. To the best knowledge of the authors, this paper presents the first spatiotemporal framework that characterizes peak \ac{AoI} within an uplink \ac{IoT} network.

Throughout this paper, we adopt the following notation. Matrices and vectors are represented as upper-case and lower-case boldface letters ($\mathbf{A}$, $\mathbf{a}$), respectively. The indicator function is denoted as $\mathbbm{1}_{\{a\}}$ which equals 1 if the expression $a$ is true and 0 otherwise. In addition, $\mathbf{1}_m$ and $\mathbf{\mathcal{I}}_m$ denote, respectively, an all ones vector and matrix of dimension $m \times m$. In addition, identity matrix of dimension $m$ is represented via $\mathbf{I}_{m}$. The complement operator is denoted by the over-bar (i.e., $\bar{v} = 1-v$).  The notations $\mathbb{P}\{\cdot\}$ and $\mathbb{E}\{\cdot\}$ denote the probability of an event and its expectation. 

The rest of the paper is organized as follows. Section \ref{sec:system_model} provides the system model, the underlying physical and \ac{MAC} assumptions, and the peak peak \ac{AoI} characterization. Section \ref{sec:sg_analysis} shows the macroscopic inter-device queueing interactions in terms of mutual interference. The proposed queueing model along with the microscopic intra-device interactions among the queues are presented in Section \ref{sec:QT_anaylsis}. Simulation results are presented in Section \ref{sec:simulation_results} and Section \ref{sec:Conclusion} summarizes the work.
	\thispagestyle{empty}
\section{System Model}\label{sec:system_model}
\subsection{PHY layer parameters}
Throughout this paper, the \acp{BS} are spatially distributed according to $\mathrm{\Phi}$ with spatial intensity $\lambda$ BS/km$^2$. The \ac{IoT} devices point process $\mathrm{\Psi}$ is constructed such that  within the Voronoi cell of every BS $b_i\in\mathrm{\Phi}$, a device is dropped uniformly and independently. Single antennas are employed at all devices and \acp{BS}. An unbounded path-loss propagation model is adopted such that the signal power attenuates at the rate $r^{-\eta}$, where $r$ is the distance and $\eta > 2$ is the path-loss exponent. Small-scale fading is assumed to be multi-path Rayleigh fading, where the intended and interference channel power gains $h$ and $g$, respectively, are exponentially distributed with unity power gain. All channel gains are assumed to be spatially and temporally independent and identical distributed (i.i.d.). Fractional path-loss inversion power control is considered at the devices with compensation factor $\epsilon$. Accordingly, the transmit power of an UE located $r$ meters away from its serving BS is given by $\rho r^{\eta\epsilon}$, where $\rho$ is a power control parameter to adjust the average received power at the BS \cite{ElSawy2014}. 

\subsection{MAC layer parameters}
The proposed framework considers a synchronized and time slotted system, in which a new packet is generated at a generic device based on an i.i.d. Bernoulli traffic generation model, with per-slot inter-packet arrival probability of $\alpha\in(0,1]$. A first-come first-serve queue is considered at each device, where failed packets are persistently retransmitted till successful reception. A packet is successfully decoded at its serving \ac{BS} if the received \ac{SIR} is larger than a predefined threshold $\theta$. In case of successful decoding, that serving \ac{BS} transmits an ACK through an error-free channel so the device can remove this intended packet from its respective queue. In case of failed decoding, the serving \ac{BS} sends out an NACK and the packet remains at the head of the device's queue, awaiting a new transmission attempt in the next time slot.
\begin{figure}
	\begin{center}
		\input{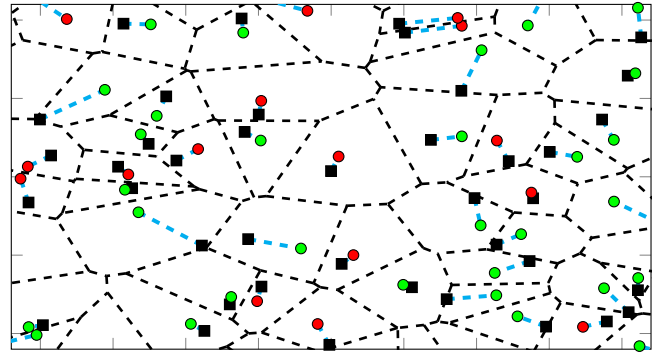}		
		\caption{A network realization for $\theta = 1$ and $\alpha=0.25$ packets/slot. Black squares depict the BSs while green and red circles represent idle and active IoT devices, respectively. The Voronoi cells of the BSs are denoted by the dashed black lines while the dashed cyan lines denote the associations of the devices to their serving BSs.}
		\label{fig:spatiotemporal_realization}
	\end{center}
\end{figure}
In Fig. \ref{fig:spatiotemporal_realization}, a snapshot realization of the network is shown. At a given time slot two different panoramas of devices can be observed; i) non-active devices (i.e., devices with empty queues) ii) devices with packets in their queues due to either a packet arrival event or failed transmission attempts of  backlogged packets\footnote{To analyze the location-dependent performance of the network, we consider a static network where for a generic network realization, $\mathrm{\Phi}$ remains static over sufficiently large number of time slots, while channel fading, queue states, and device activities change from one time slot to another.}.
\subsection{Age of Information}\label{sec:AoI}
As mentioned in Section \ref{section:introduction}, \ac{AoI} quantifies the freshness (i.e., timeliness) of information transmitted by the devices within the network \cite{Kaul2012}. For the considered time slotted system and a typical link, the \ac{AoI} $\Delta_o(t)$, tracks the \ac{AoI} evolution with time as shown in Fig. \ref{fig:AoI}. Assume that the ith packet is generated at time $Y_o(t)$, then $\Delta_o(t+1)$ is computed recursively as
\begin{equation}
\Delta_o(t+1)=\left\{\begin{array}{ll}{\Delta_o(t)+1,} & {\text { transmission failure, }} \\ {t-Y_o(t)+1,} & {\text { otherwise }}\end{array}\right.
\end{equation}
Through this paper, we consider the peak \ac{AoI} which is defined as the value of age achieved immediately before receiving the $i$-th update \cite{Huang2015}. To this end, conditioned on a fixed, yet generic spatial realization, the peak \ac{AoI}, as observed from Fig. \ref{fig:AoI}, is computed as
\begin{equation}\label{eq:peakAoI}
\mathbb{E}\{\Delta^p|  \mathrm{\Phi}\} = \mathbb{E}^{!}\Big\{\mathcal{I}_o + \mathcal{W}_o |  \mathrm{\Phi} \Big\},
\end{equation}
where $\mathbb{E}^{!}\{.\}$ is the reduced Palm expectation \cite{Haenggi2012}, $\mathcal{I}_o$ and $\mathcal{W}_o$ denote the inter-arrival time between consecutive packets and the waiting of a generic packet time in the queue, respectively. In order to characterize the peak \ac{AoI}, one must compute the  waiting time of a generic packet within the queue, which depends on, among other parameters, the traffic model, queue distribution and network-wide aggregate interference. In the remaining of this paper, we will provide a spatiotemporal framework that characterizes the peak \ac{AoI}.
\begin{figure}
	\begin{center}
		\input{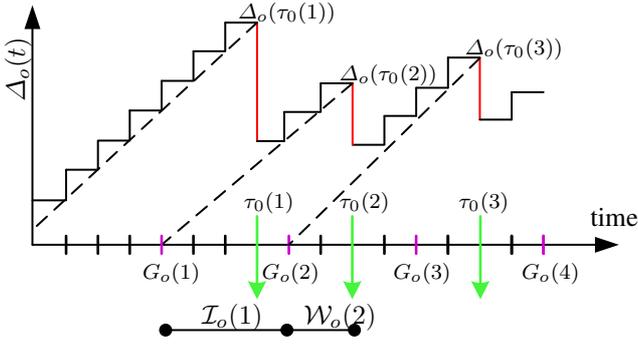}		
		\caption{AoI evolution of a typical link. The time stamps $G_o(n)$ and $\tau_o(n)$ denote the time at which the $n$-th packet was generated and successfully delivered. $\mathcal{I}_o(1)$ and $\mathcal{W}_o(2)$ denote the inter-arrival time and the waiting times.}
		\label{fig:AoI}
	\end{center}
\end{figure}
	\thispagestyle{empty}
	\section{Macroscopic Large Scale Analysis}\label{sec:sg_analysis}
Through this section, the network-wide aggregate interference will be characterized. Focusing on a fixed arbitrary spatial realization of $\mathrm{\Phi}$ and $\mathrm{\Psi}$, let $u_o \in \mathrm{\Phi}$ and $b_o = \text{argmin}_{b_o\in \mathrm{\Psi}} ||u_o-b_o||$ be a randomly selected UE and its serving BS, respectively. The transmission success probability is defined as 

\begin{align}\label{eq:PS}
P_s(\theta) &= \mathbb{P}^{ !}\left\{\frac{P_{o} h_{o}\left\|u_{o}-b_{o}\right\|^{-\eta}}{\sum_{u_{i} \in \mathrm{\Phi} \backslash u_{0}} a_{i} P_{i} g_{i}\left\|u_{i}-b_{o}\right\|^{-\eta}}>\theta | \mathrm{\Phi}, \mathrm{\Psi} \right\}, \nonumber \\
&\myeqa \prod_{r_i\in \Phi_o}  \mathbb{E}^{!} \bigg[ \bigg(\frac{1}{1 + \frac{a_{i}\theta P_ir_o^{\eta(1-\epsilon)}}{\rho r_i^{\eta}}} \bigg)\bigg| \mathrm{\Phi}, \mathrm{\Psi} \Big], \nonumber \\
&\myeqb\prod_{r_i \in \Phi_o} \Big(\frac{\bar{\chi}}{1 + \frac{\theta P_ir_o^{\eta(1-\epsilon)}}{\rho r_i^{\eta}}} + {\chi} \Big)
\end{align}
where $\mathbb{P}^{ !} \{\cdot\}$ is the reduced Palm probability, $P_o$ is the uplink transmit power of the typical device, $r_o=||u_o-b_o||$, $\mathrm{\Phi}_o=\{||b_o- \{\mathrm{\Phi}\backslash u_o\}||\}$ represents the set containing distances between the interfering devices and \ac{BS} of interest. In addition, $a_i$ and $P_i$ denote the $i$-th interference device activity profile and its uplink transmit power, respectively. The step $(a)$ follows from the exponential distribution of the channel gains (i.e., $h_o$ and $h_i$). Furthermore, the step $b$ follows as $a_i \sim \text{Bernoulli}(\alpha)$. Finally, $\bar{\chi}$ is the spatially averaged idle probability, which is dependent on the queue distribution.

Spatially, the transmission success probabilities are random variables that are location dependent. The meta distribution characterizes the transmission success probabilities variation across the network \cite{Haenggi_meta, Haenggi_meta2} as 
\begin{equation}
\bar{F}(\theta, \delta)=\mathbb{P}^{!}\{P_{s}(\theta)>\delta | \mathrm{\Phi}, \mathrm{\Psi}\},
\end{equation}
where $\delta$ is the percentile of devices within the network that can achieve $P_s(\theta)$. For uplink transmission with fractional power control, the conditional distribution of the transmission success probabilities, where the beta approximation was utilized \cite{ElSawy_meta}, \cite{Haenggi_meta} to approximate the meta distribution $f_{P_s}(\omega)$ as
\begin{align}\label{eq:metaPDF}
	F(\theta, \alpha) \approx I_{\alpha}\left(\frac{M_{1}\left(M_{1}-M_{2}\right)}{\left(M_{2}-M_{1}^{2}\right)}, \frac{\left(1-M_{1}\right)\left(M_{1}-M_{2}\right)}{\left(M_{2}-M_{1}^{2}\right)}\right), 
\end{align}
where $I_{\alpha}(a, b)=\int_{0}^{\alpha} t^{a-1}(1-t)^{b-1} \mathrm{d} t$ is the regularized
incomplete beta function, $M_1$ and $M_2$ are the first and second moments of $P_s$. The approximations of $P_s$ moments are given by the following lemma.
\begin{lemma}\label{lem:meta_lemma}
	The moments of the transmission success probabilities in an uplink network with steady state queue idle probability $\chi$, \ac{SIR} threshold $\theta$, path-loss exponent $\eta$, and fractional path-loss inversion compensation factor $\epsilon$, are approximated by $\tilde{M}_b$ given in eq.(\ref{eq:lemma1}), where $\chi$ is the residual interference intensity seen by the link and $\gamma(a,y) = \int_{0}^{b} t^{a-1}\text{e}^{-t}dt$ is the lower incomplete gamma function. 
\end{lemma}
\begin{proof}
	See Appendix \ref{se:Appendix_A}
\end{proof}
\begin{figure*}[!t]
	\normalsize
	\begin{align}\label{eq:lemma1}
\tilde{M}_{b}&= \int_{0}^{\infty} \exp \Bigg\{-z- \frac{2z^{1-\epsilon}}{\eta } \int_{\mathbbm{1}\{\epsilon = 1\}}^{\infty} y^{\frac{2}{\eta}-1}\big( 1-\Big(\frac{y+\theta\chi}{y + \theta}\Big)^b\big) \gamma\Big(1+\epsilon, zy^{\frac{2}{\eta(1-\epsilon)}}\Big) dy\Big)\Bigg\} dz, \\ 
&\myeqc \; \exp \Bigg\{-\frac{2}{\eta} \sum_{n=1}^{b}\left(\begin{array}{c}{b} \\ {n}\end{array}\right) \frac{(-1)^{n+1} (\bar{\chi}\theta)^{n}}{n-\frac{2}{\eta}}2 F_{1}\left(n, n-\frac{2}{\eta}, n+1-\frac{2}{\eta},-\theta\right) \Bigg\}. \nonumber
	\end{align}
	\hrulefill
\end{figure*}
It is observed from Lemma \ref{lem:meta_lemma} that the macroscopic network-wide aggregate characterization of the network depends on the  queues evolution, via $\chi$. In order to capture such interdependency, we first resort to the discretization of the meta distribution. Categorizing the devices within the network into classes is not feasible due to the continuous support of $P_s(\theta) \in(0,1)$, which will lead to infinite number of classes. For practicality, the distribution in (\ref{eq:metaPDF}) is quantized into $N$ \ac{QoS} classes based on the importance criteria \cite{Chisci2019}. The network categorization process of the distribution in (\ref{eq:metaPDF}) for the $n$-th class is based first on computing $\omega_n$ as 
\begin{equation}\label{eq:pmf1}
F_{P_s}(\omega_n)-F_{P_s}(\omega_{n+1})=\int_{\omega_n}^{\omega_{n+1}} f_{P_s}(\omega) d\omega=\frac{1}{N},
\end{equation}
where $n\in[1,2,\cdots N]$. Afterwards, the discrete probability mass function $d_n$ (i.e., $\mathbb{P}\{P_s = d_n\} = \frac{1}{N}$) can be evaluated using Bisection method as 
\begin{equation}\label{eq:pmf2}
\int_{\omega_n}^{d_n} f_{P_s}(\omega)d\omega = \int_{d_n}^{\omega_{n+1}} f_{P_s}(\omega)d\omega.
\end{equation} 
The computation of $d_n,\; \forall n=[1,2,\cdots,N]$ via eq. (\ref{eq:pmf1}) and (\ref{eq:pmf2}) quantizes the meta distribution into $N$ equiprobable classes as shown in Fig. \ref{fig:quant_meta}. The queue's service rate of a device belonging to the $n$-class is determined by  $d_n$. Finally, in order to fully evaluate $d_n$, one needs first to compute $\chi$ which is dependent on the queue characteristics as presented in the next section. 

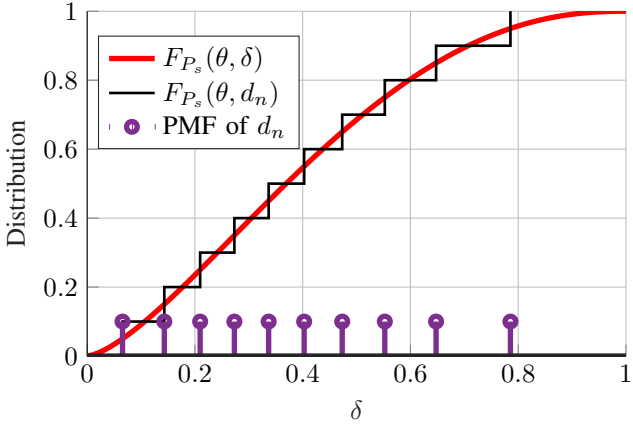
\begin{figure} 
	\centering
	\definecolor{mycolor1}{rgb}{0.49020,0.18039,0.56078}%

\begin{tikzpicture}

\begin{axis}[%
width=0.8\columnwidth,
height=1.8in,
scale only axis,
xmin=0,
xmax=1,
xlabel style={font=\color{white!15!black}},
xlabel={$\delta$},
ymin=0,
ymax=1,
ytick={0,0.2,0.4,0.6,0.8,1},
ylabel style={font=\color{white!15!black}},
ylabel={Distribution},
axis x line*=bottom,
axis y line*=left,
xmajorgrids,
ymajorgrids,
legend style={at={(0.02,0.6)}, anchor=south west, legend cell align=left, align=left, draw=white!15!black}
]
\addplot [color=red, line width=2.0pt]
table[row sep=crcr]{%
	0	0\\
	0.01	0.00316147369524374\\
	0.02	0.00882172754478877\\
	0.03	0.0160243637660381\\
	0.04	0.0244154579543246\\
	0.05	0.0337837571741215\\
	0.06	0.043982309739499\\
	0.07	0.0549000574866723\\
	0.08	0.0664485383584482\\
	0.09	0.0785546785910195\\
	0.1	0.0911564860060308\\
	0.11	0.104200290013568\\
	0.12	0.117638876499696\\
	0.13	0.131430173837579\\
	0.14	0.145536295509061\\
	0.15	0.159922822930666\\
	0.16	0.174558255558411\\
	0.17	0.18941358080871\\
	0.18	0.204461931890115\\
	0.19	0.219678311495064\\
	0.2	0.235039365739724\\
	0.21	0.250523197062651\\
	0.22	0.266109207764718\\
	0.23	0.28177796795935\\
	0.24	0.297511103195023\\
	0.25	0.313291198098395\\
	0.26	0.329101713189171\\
	0.27	0.344926912619295\\
	0.28	0.360751801045421\\
	0.29	0.376562068193998\\
	0.3	0.39234403995015\\
	0.31	0.40808463501454\\
	0.32	0.423771326340905\\
	0.33	0.439392106701285\\
	0.34	0.454935457833978\\
	0.35	0.470390322716766\\
	0.36	0.485746080579267\\
	0.37	0.500992524326858\\
	0.38	0.516119840096992\\
	0.39	0.531118588708916\\
	0.4	0.54597968880141\\
	0.41	0.560694401481357\\
	0.42	0.575254316329802\\
	0.43	0.589651338632364\\
	0.44	0.603877677718117\\
	0.45	0.617925836305845\\
	0.46	0.631788600769273\\
	0.47	0.645459032243929\\
	0.48	0.65893045850784\\
	0.49	0.672196466576648\\
	0.5	0.685250895961113\\
	0.51	0.698087832541443\\
	0.52	0.710701603018723\\
	0.53	0.723086769908891\\
	0.54	0.73523812704939\\
	0.55	0.74715069559294\\
	0.56	0.758819720466763\\
	0.57	0.770240667279302\\
	0.58	0.781409219659894\\
	0.59	0.792321277020183\\
	0.6	0.802972952729238\\
	0.61	0.813360572697503\\
	0.62	0.82348067436786\\
	0.63	0.833330006115298\\
	0.64	0.84290552706002\\
	0.65	0.852204407302355\\
	0.66	0.861224028591577\\
	0.67	0.869961985444899\\
	0.68	0.878416086737417\\
	0.69	0.886584357788906\\
	0.7	0.894465042979161\\
	0.71	0.902056608930233\\
	0.72	0.909357748301663\\
	0.73	0.91636738425389\\
	0.74	0.923084675645824\\
	0.75	0.929509023045427\\
	0.76	0.9356400756478\\
	0.77	0.941477739214249\\
	0.78	0.947022185169238\\
	0.79	0.952273861021269\\
	0.8	0.957233502310285\\
	0.81	0.961902146330657\\
	0.82	0.96628114793847\\
	0.83	0.970372197829431\\
	0.84	0.974177343776026\\
	0.85	0.977699015449496\\
	0.86	0.980940053638466\\
	0.87	0.98390374493423\\
	0.88	0.986593863317886\\
	0.89	0.989014720613508\\
	0.9	0.99117122855845\\
	0.91	0.993068976449822\\
	0.92	0.994714330249099\\
	0.93	0.996114562224595\\
	0.94	0.997278025818323\\
	0.95	0.998214400924384\\
	0.96	0.998935056169969\\
	0.97	0.999453623675835\\
	0.98	0.999787013265266\\
	0.99	0.999957554480932\\
	1	1\\
};
\addlegendentry{$F_{P_s}(\theta, \delta)$}

\addplot [color=black, line width=1.0pt]
table[row sep=crcr]{%
	-inf	0\\
	0.0655875205993652	0\\
	0.0655875205993652	0.1\\
	0.14312219619751	0.1\\
	0.14312219619751	0.2\\
	0.209662914276123	0.2\\
	0.209662914276123	0.3\\
	0.273205280303955	0.3\\
	0.273205280303955	0.4\\
	0.336818218231201	0.4\\
	0.336818218231201	0.5\\
	0.402720928192139	0.5\\
	0.402720928192139	0.6\\
	0.473352909088135	0.6\\
	0.473352909088135	0.7\\
	0.552421092987061	0.7\\
	0.552421092987061	0.8\\
	0.647600650787354	0.8\\
	0.647600650787354	0.9\\
	0.785600185394287	0.9\\
	0.785600185394287	1\\
	inf	1\\
};
\addlegendentry{$F_{P_s}(\theta, d_n)$}

\addplot[ycomb, color=mycolor1, line width=2.0pt, mark=o, mark options={solid, mycolor1}] table[row sep=crcr] {%
	0.0655875205993652	0.1\\
	0.14312219619751	0.1\\
	0.209662914276123	0.1\\
	0.273205280303955	0.1\\
	0.336818218231201	0.1\\
	0.402720928192139	0.1\\
	0.473352909088135	0.1\\
	0.552421092987061	0.1\\
	0.647600650787354	0.1\\
	0.785600185394287	0.1\\
};
\addplot[forget plot, color=white!15!black, line width=2.0pt] table[row sep=crcr] {%
	0	0\\
	1	0\\
};
\addlegendentry{$\text{PMF of }d_n$}

\end{axis}
\end{tikzpicture}%
	\caption{Quantized meta distribution for $N=10$, hypothetical $\chi = 0.5$, $\theta = 5$ dB, $\eta = 4$ and $\epsilon=1$.}
	\label{fig:quant_meta} 
\end{figure} 
	\thispagestyle{empty}
\section{microscopic queueing theory analysis}\label{sec:QT_anaylsis}
The mathematical model for the microscopic scale (i.e., queue evolution) at the devices will be presented in this section. In the proposed framework, the device's location-dependency is captured via its departure probability (i.e. QoS class dependent), which is static over a sufficiently large number of time slots. However , the departure probability for each queue is still random and independent from one time slot to another due to the independent random variations of the channel gains and activity profiles of the interfering devices. Consequently,  a Geo/Geo/1 queueing model is adopted to track the queue evolution at each device, where Geo stands for geometric inter-arrival and departure processes. It is important to note that the geometric departure is an approximation that capitalizes on the negligible temporal correlation of the departure probabilities once the location-dependent QoS class is determined \cite{Chisci2019}. To this end, the queue transitions for a device within the $n$-th class are captured through the following probability transition matrix
\begin{equation}\label{eq:QBD}
\mathbf{P}_n=\left[\begin{array}{lllll}{\bar{\alpha}} & {\alpha} & {} & {} & {} \\ {\bar{\alpha}d_n} & {\alpha d_n + \bar{\alpha}\bar{d}_n} & {\alpha \bar{d}_n} & {} & {} \\ {} & {\bar{\alpha}d_n} & {\alpha d_n + \bar{\alpha}\bar{d}_n} & {\alpha \bar{d}_n} & {} \\ {} & {} & {\ddots} & {\ddots} & {\ddots}\end{array}\right].
\end{equation}
To analyze the queue's stability, one is interested to determine the critical arrival rate after which the probability of having unstable queues starts to dominate and the queue's idle probability is zero  \cite{Loynes1962}. Mathematically, for the \ac{DTMC} in (\ref{eq:QBD}) to be stable, the equality $\frac{\alpha}{d_n} < 1$ must be satisfied. For unstable \acp{DTMC}, the idle probability is naturally 0. To this end, let $\mathbf{x}_n = [x_{0,n} \; x_{1,n} \; x_{2,n} \; \cdots]$ be the steady state probability vector of the $n$-th class, where $x_{i,n}$ is the probability that a device belonging to the $n$-th class has $i$ packets. The idle probability of device in the $j-th$ class is evaluated as \cite{Alfa2015}
\begin{equation}\label{eq:x_0}
x_{i,n} = R_n^i\frac{x_{0,n}}{\bar{d}_n}, \;\; \text{where } R_n = \frac{\alpha \bar{d}_n}{\bar{\alpha}d_n}, \;\; \text{and } x_{0,n} = \frac{d_n - \alpha}{d_n}.
\end{equation}
Resorting to the mean field theory, the spatially averaged idle probability $\chi$ that is required to evaluate $F(\theta, \delta)$, is computed by averaging over the $N$ classes temporal idle probabilities as 
\begin{equation}\label{eq:res}
\chi = \frac{1}{N}\sum_{n=1}^{N} x_{0,n}.
\end{equation}
Once the queue distribution is characterized, one can proceed with evaluating the temporal distribution of a generic packet within the considered Geo/Geo/1 queue, which is the major component in computing the peak \ac{AoI} as explained in Section \ref{sec:AoI}. 
As mentioned earlier, let $\mathcal{W}_n$ be the waiting time in the queue for a packet at a device belonging to the $n$-th class and $\mathcal{W}_o^m = \mathbb{P}\{\mathcal{W}_o = m\}$. The waiting time for the $n$-th class is 
\begin{equation}\label{eq:waiting}
\mathcal{W}_{n}^m= \begin{cases}
\frac{d_n -\alpha}{d_n},  &  m = 0 \\ 
\sum_{v=1}^{i} x_{v,n}\left(\begin{array}{c}{i-1} \\ {v-1}\end{array}\right) d_n^{v}(1-d_n)^{i-v} &  m \ge 1.
\end{cases}
\end{equation}
Finally, by plugging eq.(\ref{eq:waiting}) in eq.(\ref{eq:peakAoI}), the network-wide peak \ac{AoI} of a device within the $n$-th class is computed as 
\begin{equation}\label{eq:AoI_wait}
\mathbb{E}\{\Delta^p|  \mathrm{\Phi}\} = \frac{1}{\alpha} + \frac{1}{N}\Big(\sum_{\varrho=1}^{N}\sum_{k=0}^{\infty} \mathcal{W}_\varrho^k k\Big). 
\end{equation}
It is clear that to evaluate the peak \ac{AoI}, the queue distribution is required, which in turn is dependent on the network-wide aggregate interference, via $d_n$. To solve such interdependency, Algorithm 1 is presented which converges uniquely to a solution by virtue of fixed point theorem \cite{Zhou2016}. 
\begin{figure}
	\begin{algorithm}[H]\label{alg:iterative}
		\caption{Iterative computation of ${F}(\theta, \delta)$ and $\chi$}\label{euclid}
		\begin{algorithmic}
			\Procedure{}{$\alpha, \eta, \epsilon, \theta, N, \varphi$} 
			\LState initialize  $\chi$ 
			\While {$||\chi^k - \chi^{k-1}|| \geq \varphi$}  
			\LState Compute the moments $\tilde{M}_{b}$ from Lemma \ref{lem:meta_lemma}				 
			\LState Evaluate $F(\theta, \alpha)$ based on (\ref{eq:metaPDF})
			\LState Compute $d_i,\forall i=[1\cdots N]$ from the Discretized 
			\LState $F(\theta, \alpha)$ based on (\ref{eq:pmf1}) and (\ref{eq:pmf2}) 
			\For {$n=[1,2,\cdots,N]$} 
			\If {$\alpha < d_n$}
			\LState Compute $x_{0,n}$ based on (\ref{eq:x_0})
			\Else
			\LState Set $x_{(0,n)} = 0$
			\EndIf
			\EndFor
			\LState Compute $\chi$ based on (\ref{eq:res})
			\LState Increment k
			\EndWhile
			\LState \Return  ${F}(\theta, \delta)$ and $\chi$
			\EndProcedure
		\end{algorithmic}
	\end{algorithm}
\end{figure}

	\thispagestyle{empty}
\section{Simulation Results}\label{sec:simulation_results}
In this section, the proposed spatiotemporal framework for the Bernoulli traffic is first validated. Afterwards, insights on the temporal distribution and peak \ac{AoI} for different system parameters are demonstrated. First, framework validation is conducted via independent Monte Carlo simulations. The developed simulation framework incorporates microscopic averaging ensuring ergodicity, in which the steady state statistics of the queues employed at each device are collected. The simulation area is $10 \times 10 \text{ km}^2$ with a wrapped-around boundaries to ensure unbiased statistics imposed by the network boundary devices. In order to ensure that the queues are in a steady state, simulation is first initiated with all queues at the devices as being idle and then it runs for a sufficiently high number of time slots until the steady-state is reached. Let $\hat{x}_0^t$ denotes the average idle steady state probability for the $t$-th iteration across all the devices within the network. Mathematically, the steady state is realized once $||\hat{x}_0^k-\hat{x}_0^{k-1}|| < \varphi$, where $\varphi$ is some predetermined tolerance. After steady state is reached, all temporal statistics are then gathered based on sufficiently large number of microscopic realizations. Unless otherwise stated, we consider the following parameters: $N=10,\; \eta = 4$, $\rho=-90$, $\epsilon = 1$. 
\begin{figure} 
	\centering
	\input{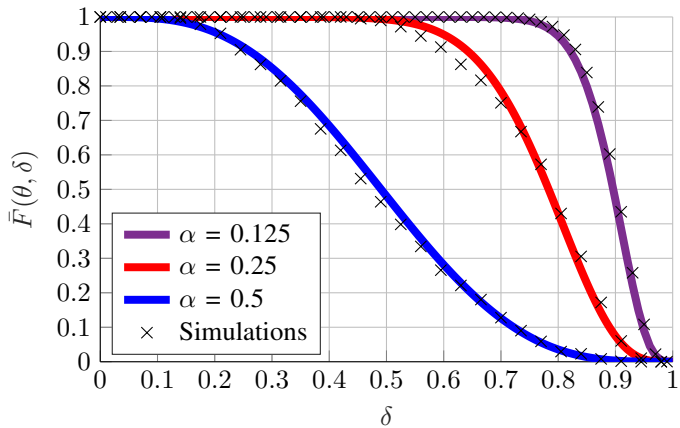}	
	\caption{Model verification for $\theta=1$. Solid lines and black marks depict analysis and simulation, respectively.}
	\label{fig:meta_sim} 
\end{figure} 

In Fig. \ref{fig:meta_sim} we show the meta distribution for $\theta=1$ and different values of $\alpha$. One can observe a close match between the simulation and the proposed framework, which implies that the iterative solution presented in Algorithm 1 in Section \ref{sec:QT_anaylsis} can capture the interdependency between the network-wide aggregate interference and the queues evolution. As the $\alpha$ increases, the capability of a given device to dispatch its packet generated at a given time slot deteriorates, due to the increased channel access and transmission attempts within the network. Accordingly, this packet will attempt retransmissions in consecutive time slots, thus, contributing to the aggregate interference on the devices that attempt transmissions in coming time slots. Such a consequential effect of increased traffic load affects the percentile of devices within the network to achieve a given transmission success probability, as illustrated via the meta distribution.    

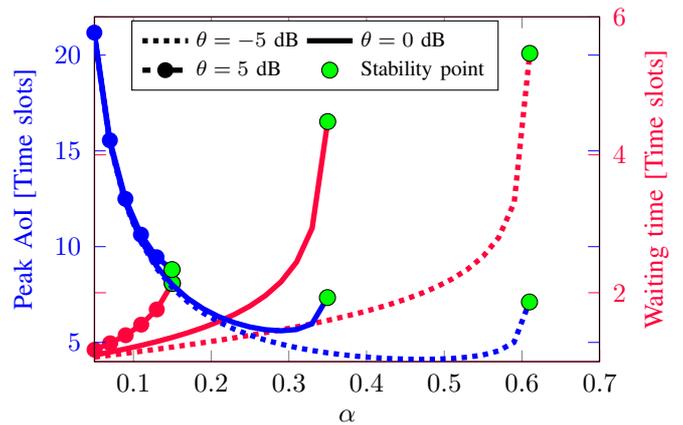
\begin{figure} 
	\centering
	\definecolor{mycolor1}{rgb}{0.0, 0.0, 1.0}
\definecolor{mycolor2}{rgb}{1.0, 0.01, 0.24}
\begin{tikzpicture}
\pgfplotsset{
	width=0.75\columnwidth,
	height=1.8in,
	scale only axis,
	xmin=0.05,
	xmax=0.7,
	legend style={legend columns =2,
		at={(0.8,0.99)},
		anchor=north east
	},
	legend cell align=left,
	ticklabel shift={0.05cm},
	tick label style={/pgf/number format/1000 sep=}
}
\begin{axis}[%
every outer y axis line/.append style={mycolor2},
every y tick label/.append style={font=\color{mycolor2}},
every y tick/.append style={mycolor2},
ymin=1,
ymax=6,
ylabel style={font=\color{mycolor2}},
ylabel={Waiting time [Time slots]},
xlabel={$\alpha$},
yticklabel pos=right
]
\addlegendimage{legend columns=2, style={color=black, dotted, line width=2.0pt]}}
\addlegendentry{\footnotesize$\theta = -5 \text{ dB}$}
\addlegendimage{style={color=black, line width=2.0pt}}
\addlegendentry{\footnotesize$\theta = 0 \text{ dB}$}
\addlegendimage{style={color=black,dash pattern=on 3pt off 6pt on 6pt off 6pt,
		mark=*, mark options={solid}, line width=2.0pt}}
\addlegendentry{\footnotesize$\theta = 5 \text{ dB}$}

\addplot [color=black, only marks, mark size=3pt, mark=*, mark options={solid, fill=green}]
table[row sep=crcr]{%
	0.15	2.13601130109271\\
};\addlegendentry{\footnotesize{Stability point}}

\addplot [color=mycolor2, mark=*, mark options={solid},line width=2.0pt]
table[row sep=crcr]{%
	0.05	1.17360520490003\\
	0.07	1.26722974181006\\
	0.09	1.38401587897303\\
	0.11	1.53748731610666\\
	0.13	1.75772224453968\\
	0.15	2.13601130109271\\
	0.17	inf\\
	0.19	inf\\
	0.21	inf\\
	0.23	inf\\
	0.25	inf\\
	0.27	inf\\
	0.29	inf\\
	0.31	inf\\
	0.33	inf\\
	0.35	inf\\
	0.37	inf\\
	0.39	inf\\
	0.41	inf\\
	0.43	inf\\
	0.45	inf\\
	0.47	inf\\
	0.49	inf\\
	0.51	inf\\
	0.53	inf\\
	0.55	inf\\
	0.57	inf\\
	0.59	inf\\
	0.61	inf\\
	0.63	inf\\
	0.65	inf\\
	0.67	inf\\
	0.69	inf\\
	0.71	inf\\
	0.73	inf\\
	0.75	inf\\
	0.77	inf\\
	0.79	inf\\
};
\addplot [color=black, only marks, mark size=3pt, mark=*, mark options={solid, fill=green}]
table[row sep=crcr]{%
	0.35	4.48216917628451\\
};
\addplot [color=black, only marks, mark size=3pt, mark=*, mark options={solid, fill=green}]
table[row sep=crcr]{%
	0.15	2.13601130109271\\
};

\addplot [color=mycolor2, line width=2.0pt]
table[row sep=crcr]{%
	0.05	1.09615375938418\\
	0.07	1.13935830606912\\
	0.09	1.18592187637567\\
	0.11	1.23641776944209\\
	0.13	1.29158608006148\\
	0.15	1.35235526369897\\
	0.17	1.42002093006921\\
	0.19	1.49631010569694\\
	0.21	1.58362167583047\\
	0.23	1.68563005643157\\
	0.25	1.80794772258687\\
	0.27	1.9599426257718\\
	0.29	2.15950128915614\\
	0.31	2.44583173683114\\
	0.33	2.93680803069121\\
	0.35	4.48216917628451\\
	0.37	inf\\
	0.39	inf\\
	0.41	inf\\
	0.43	inf\\
	0.45	inf\\
	0.47	inf\\
	0.49	inf\\
	0.51	inf\\
	0.53	inf\\
	0.55	inf\\
	0.57	inf\\
	0.59	inf\\
	0.61	inf\\
	0.63	inf\\
	0.65	inf\\
	0.67	inf\\
	0.69	inf\\
	0.71	inf\\
	0.73	inf\\
	0.75	inf\\
	0.77	inf\\
	0.79	inf\\
};

\addplot [color=mycolor2, dotted, line width=2.0pt]
table[row sep=crcr]{%
	0.05	1.06612918647197\\
	0.07	1.09372447480384\\
	0.09	1.12204582214216\\
	0.11	1.15115125126206\\
	0.13	1.18110523391927\\
	0.15	1.21198065279054\\
	0.17	1.24386170039379\\
	0.19	1.27684568542407\\
	0.21	1.311041837097\\
	0.23	1.34658104244248\\
	0.25	1.38360930045667\\
	0.27	1.42231570170467\\
	0.29	1.46291350035142\\
	0.31	1.50565963631341\\
	0.33	1.5508728051319\\
	0.35	1.59894494001488\\
	0.37	1.65036779796704\\
	0.39	1.70577751842421\\
	0.41	1.76600675985318\\
	0.43	1.83213814005848\\
	0.45	1.90568622713971\\
	0.47	1.98889854047503\\
	0.49	2.08497552375595\\
	0.51	2.1991781703361\\
	0.53	2.34049783957242\\
	0.55	2.5266598304001\\
	0.57	2.79998995512125\\
	0.59	3.30689162929792\\
	0.61	5.47158987551092\\
	0.63	inf\\
	0.65	inf\\
	0.67	inf\\
	0.69	inf\\
	0.71	inf\\
	0.73	inf\\
	0.75	inf\\
	0.77	inf\\
	0.79	inf\\
};

\addplot [color=black, only marks, mark size=3pt, mark=*, mark options={solid, fill=green}]
table[row sep=crcr]{%
	0.61	5.47158987551092\\
};

\end{axis}

\begin{axis}[%
every y tick label/.append style={font=\color{mycolor1}},
every y tick/.append style={mycolor1},
ymin=4,
ymax=22,
ylabel style={font=\color{mycolor1}},
ylabel={Peak AoI [Time slots]},
yticklabel pos=left,
xtick=\empty
]
\addplot [color=mycolor1, mark=*, mark options={solid}, line width=2.0pt]
table[row sep=crcr]{%
	0.05	21.1736052049\\
	0.07	15.5529440275243\\
	0.09	12.4951269900841\\
	0.11	10.6283964070157\\
	0.13	9.45002993684737\\
	0.15	8.80267796775938\\
	0.17	inf\\
	0.19	inf\\
	0.21	inf\\
	0.23	inf\\
	0.25	inf\\
	0.27	inf\\
	0.29	inf\\
	0.31	inf\\
	0.33	inf\\
	0.35	inf\\
	0.37	inf\\
	0.39	inf\\
	0.41	inf\\
	0.43	inf\\
	0.45	inf\\
	0.47	inf\\
	0.49	inf\\
	0.51	inf\\
	0.53	inf\\
	0.55	inf\\
	0.57	inf\\
	0.59	inf\\
	0.61	inf\\
	0.63	inf\\
	0.65	inf\\
	0.67	inf\\
	0.69	inf\\
	0.71	inf\\
	0.73	inf\\
	0.75	inf\\
	0.77	inf\\
	0.79	inf\\
};

\addplot [color=black, only marks, mark size=3pt, mark=*, mark options={solid, fill=green}]
table[row sep=crcr]{%
	0.15	8.80267796775938\\
};
\addplot [color=mycolor1, line width=2.0pt]
table[row sep=crcr]{%
	0.05	21.0961537593842\\
	0.07	15.4250725917834\\
	0.09	12.2970329874868\\
	0.11	10.3273268603512\\
	0.13	8.98389377236917\\
	0.15	8.01902193036564\\
	0.17	7.30237387124569\\
	0.19	6.75946800043378\\
	0.21	6.34552643773523\\
	0.23	6.03345614338809\\
	0.25	5.80794772258687\\
	0.27	5.66364632947551\\
	0.29	5.6077771512251\\
	0.31	5.67163818844404\\
	0.33	5.96711106099424\\
	0.35	7.33931203342737\\
	0.37	inf\\
	0.39	inf\\
	0.41	inf\\
	0.43	inf\\
	0.45	inf\\
	0.47	inf\\
	0.49	inf\\
	0.51	inf\\
	0.53	inf\\
	0.55	inf\\
	0.57	inf\\
	0.59	inf\\
	0.61	inf\\
	0.63	inf\\
	0.65	inf\\
	0.67	inf\\
	0.69	inf\\
	0.71	inf\\
	0.73	inf\\
	0.75	inf\\
	0.77	inf\\
	0.79	inf\\
};

\addplot [color=black, only marks, mark size=3pt, mark=*, mark options={solid, fill=green}]
table[row sep=crcr]{%
	0.35	7.33931203342737\\
};

\addplot [color=mycolor1, dotted, line width=2.0pt]
table[row sep=crcr]{%
	0.05	21.066129186472\\
	0.07	15.3794387605181\\
	0.09	12.2331569332533\\
	0.11	10.2420603421712\\
	0.13	8.87341292622696\\
	0.15	7.87864731945721\\
	0.17	7.12621464157026\\
	0.19	6.54000358016092\\
	0.21	6.07294659900176\\
	0.23	5.694407129399\\
	0.25	5.38360930045667\\
	0.27	5.12601940540837\\
	0.29	4.91118936242039\\
	0.31	4.73146608792632\\
	0.33	4.58117583543493\\
	0.35	4.45608779715774\\
	0.37	4.35307050066974\\
	0.39	4.26988008252678\\
	0.41	4.20503115009709\\
	0.43	4.15771953540732\\
	0.45	4.12790844936193\\
	0.47	4.11655811494311\\
	0.49	4.12579185028657\\
	0.51	4.15996248406159\\
	0.53	4.22729029240261\\
	0.55	4.34484164858192\\
	0.57	4.55437592003353\\
	0.59	5.00180688353521\\
	0.61	7.110934137806\\
	0.63	inf\\
	0.65	inf\\
	0.67	inf\\
	0.69	inf\\
	0.71	inf\\
	0.73	inf\\
	0.75	inf\\
	0.77	inf\\
	0.79	inf\\
};

\addplot [color=black, only marks, mark size=3pt, mark=*, mark options={solid, fill=green}]
table[row sep=crcr]{%
	0.61	7.110934137806\\
};
\end{axis}
\end{tikzpicture}%
	\caption{Peak AoI (left) and average waiting time (right) for increasing packet generation periodicity ($T$) and  $\theta$.}
	\label{fig:AoI_1} 
\end{figure} 

Moving to the \ac{AoI} in Fig. \ref{fig:AoI_1}, the peak \ac{AoI} along with average waiting time are presented. As explained in Section \ref{sec:AoI}, the peak \Ac{AoI} depends on the inter-arrival and system waiting times of a randomly selected packet within the queue. First we investigate the effect of $\theta$. As $\theta$ increases, packets are subjected to a more stringent requirement on their achieved SIR. This leads to increased retransmissions, thus, increasing the mutual interference due to lower idle probabilities. The increased mutual interference hinders the successful departure of the packets from their respective queues and lead to queue instability in some devices, yielding infinite waiting times and peak \ac{AoI}. Moving to the arrival probability, for low values of $\alpha$, the inter-arrival component dominates, yielding high values of peak \ac{AoI}, while the waiting time is low. As $\alpha$ increases, the waiting times dominates, yielding an increase in the peak \ac{AoI} till point of queue instability, as indicated by the stability point. The effect of $\theta$ on the stability frontiers can be explained in a similar fashion to that of Fig. \ref{fig:meta_sim}, where increasing $\theta$ diminishes the stability region due to the increased network-wide aggregate interference. .  

Finally, Fig. \ref{fig:AoI_2} presents the peak \ac{AoI} among the different \ac{QoS} classes within the network. The shown classes are sorted in an ascending order with respect to $d_n$.  For $\alpha=0.05$, the inter-arrival times dominates the peak \ac{AoI}, leading to a nearly-constant peak \ac{AoI} over all the classes. The location-dependency is more clear for $\alpha=0.15$ and $\alpha=0.25$. Consequently, classes with lower indices experience large peak \ac{AoI} due to their larger waiting times (i.e., effect of the location dependency captured via the meta distribution). For large traffic load (i.e., $\alpha = 0.25$), all except last two classes are unstable, which lead to infinite peak \ac{AoI}.
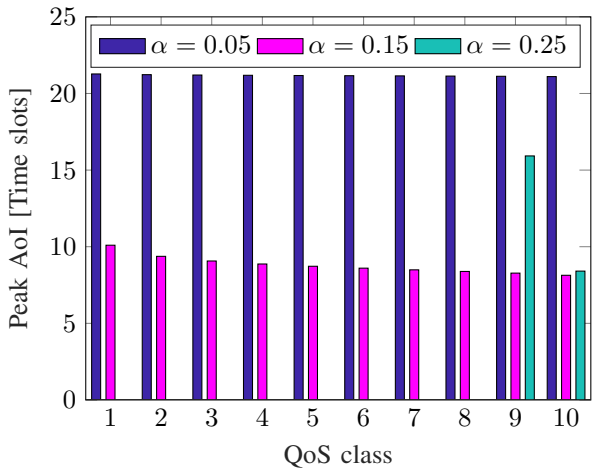
\begin{figure}
	\centering
	\definecolor{mycolor1}{rgb}{0.24220,0.15040,0.66030}%
\definecolor{mycolor2}{rgb}{1,0.0,1.0}%
\definecolor{mycolor3}{rgb}{0.09640,0.75000,0.71204}%
\definecolor{mycolor4}{rgb}{0.69, 0.4, 0.0}%
\definecolor{mycolor5}{rgb}{0.97690,0.98390,0.08050}%

\begin{tikzpicture}
\begin{axis}[%
width=0.75\columnwidth,
height=2in,
scale only axis,
bar shift auto,
xmin=0.511111111111111,
xmax=10.4888888888889,
xtick={ 1,  2,  3,  4,  5,  6,  7,  8,  9, 10},
xlabel style={font=\color{white!15!black}},
xlabel={QoS class},
ymin=0,
ymax=25,
ylabel style={font=\color{white!15!black}},
ylabel={Peak AoI [Time slots]},
axis background/.style={fill=white},
legend style={legend columns=-1, legend cell align=left, align=left, draw=white!15!black}
]
\addplot[ybar, bar width=0.178, fill=mycolor1, draw=black, area legend] table[row sep=crcr] {%
	1	21.2741981229096\\
	2	21.2281323270725\\
	3	21.204349022023\\
	4	21.1871202467461\\
	5	21.172904545306\\
	6	21.1602155868289\\
	7	21.1481502273588\\
	8	21.135893745752\\
	9	21.1221916106923\\
	10	21.1028966143111\\
};
\addplot[forget plot, color=white!15!black] table[row sep=crcr] {%
	0.511111111111111	0\\
	10.4888888888889	0\\
};
\addlegendentry{$\alpha = 0.05$}

\addplot[ybar, bar width=0.178, fill=mycolor2, draw=black, area legend] table[row sep=crcr] {%
	1	10.0982823481385\\
	2	9.37216774257099\\
	3	9.06839059465288\\
	4	8.87241488392602\\
	5	8.72407469905245\\
	6	8.60087953891906\\
	7	8.49116007365773\\
	8	8.38655384498658\\
	9	8.2771907465243\\
	10	8.13566520516527\\
};
\addplot[forget plot, color=white!15!black] table[row sep=crcr] {%
	0.511111111111111	0\\
	10.4888888888889	0\\
};
\addlegendentry{$\alpha = 0.15$}

\addplot[ybar, bar width=0.178, fill=mycolor3, draw=black, area legend] table[row sep=crcr] {%
	1	inf\\
	2	inf\\
	3	inf\\
	4	inf\\
	5	inf\\
	6	inf\\
	7	inf\\
	8	inf\\
	9	15.9246065812388\\
	10	8.40814814680575\\
};
\addplot[forget plot, color=white!15!black] table[row sep=crcr] {%
	0.511111111111111	0\\
	10.4888888888889	0\\
};
\addlegendentry{$\alpha = 0.25$}

\end{axis}
\end{tikzpicture}%
	\caption{Peak AoI for $N=10$ QoS classes and $\theta = 5$ dB.}
	\label{fig:AoI_2} 
\end{figure} 
	\thispagestyle{empty}
\section{Conclusion}\label{sec:Conclusion}

We present a tractable spatiotemporal mathematical framework to characterize the peak \ac{AoI} in uplink \ac{IoT} networks. We leverage tools from stochastic geometry to analyze the location-dependent performance of the network under Bernoulli traffic via the meta distribution. In addition, tools from queueing theory are utilized to track the queues evolution at each device and derive interference-aware expressions for the queue distribution and packets delay. To this end, an iterative algorithm is presented to solve the developed spatiotemporal model. Expressions for the packets waiting distribution within the queue and the peak \ac{AoI} are provided. Simulation results are presented to validate the proposed spatiotemporal framework for different traffic load scenarios. The peak \ac{AoI} along with the stability frontiers of the network are presented. It is shown that the peak AoI  highly depends on the spatial location, traffic arrival load, and decoding threshold. 
	\thispagestyle{empty}
\bibliographystyle{./lib/IEEEtran.cls}
\bibliography{./literature/Literature_Local}

\begin{thebibliography}{10}
\providecommand{\url}[1]{#1}
\csname url@samestyle\endcsname
\providecommand{\newblock}{\relax}
\providecommand{\bibinfo}[2]{#2}
\providecommand{\BIBentrySTDinterwordspacing}{\spaceskip=0pt\relax}
\providecommand{\BIBentryALTinterwordstretchfactor}{4}
\providecommand{\BIBentryALTinterwordspacing}{\spaceskip=\fontdimen2\font plus
\BIBentryALTinterwordstretchfactor\fontdimen3\font minus
  \fontdimen4\font\relax}
\providecommand{\BIBforeignlanguage}[2]{{%
\expandafter\ifx\csname l@#1\endcsname\relax
\typeout{** WARNING: IEEEtran.bst: No hyphenation pattern has been}%
\typeout{** loaded for the language `#1'. Using the pattern for}%
\typeout{** the default language instead.}%
\else
\language=\csname l@#1\endcsname
\fi
#2}}
\providecommand{\BIBdecl}{\relax}
\BIBdecl

\bibitem{3GPP2018}
3GPP, ``{TS} 22.261 service requirements for next generation new services and
  markets,'' \emph{3rd Generation Partnership Project (3GPP), v16.8.0}, 2019,.

\bibitem{Kim2012}
K.~{Kim} and P.~R. {Kumar}, ``Cyber–physical systems: A perspective at the
  centennial,'' \emph{Proceedings of the IEEE}, vol. 100, no. Special
  Centennial Issue, pp. 1287--1308, May 2012.

\bibitem{NGMNA2016}
NGMNA, ``Recommendations for {NGMN KPIs} and requirements for 5{G},''
  \emph{Next Generation Mobile Networks Alliance}, 2016.

\bibitem{Kaul2012}
S.~{Kaul}, R.~{Yates}, and M.~{Gruteser}, ``Real-time status: How often should
  one update?'' in \emph{2012 Proceedings IEEE INFOCOM}, March 2012, pp.
  2731--2735.

\bibitem{Yates2019}
R.~D. {Yates} and S.~K. {Kaul}, ``The age of information: Real-time status
  updating by multiple sources,'' \emph{IEEE Transactions on Information
  Theory}, vol.~65, no.~3, pp. 1807--1827, March 2019.

\bibitem{AoIDhillon}
M.~A. {Abd-Elmagid} and H.~S. {Dhillon}, ``Average peak age-of-information
  minimization in {UAV}-assisted {IoT} networks,'' \emph{IEEE Transactions on
  Vehicular Technology}, vol.~68, no.~2, Feb 2019.

\bibitem{Ceran2018}
E.~T. {Ceran}, D.~{Gündüz}, and A.~{György}, ``Average age of information
  with hybrid {ARQ} under a resource constraint,'' in \emph{2018 IEEE Wireless
  Communications and Networking Conference (WCNC)}, April 2018.

\bibitem{Metzger2019}
F.~{Metzger} \emph{et~al.}, ``Modeling of aggregated {IoT} traffic and its
  application to an {IoT} cloud,'' \emph{Proceedings of the IEEE}, vol. 107,
  no.~4, pp. 679--694, April 2019.

\bibitem{Palattella2016}
M.~R. {Palattella} \emph{et~al.}, ``Internet of things in the {5G} era:
  Enablers, architecture, and business models,'' \emph{IEEE Journal on Selected
  Areas in Communications}, vol.~34, no.~3, pp. 510--527, March 2016.

\bibitem{Jiang2019}
Z.~{Jiang} \emph{et~al.}, ``Timely status update in wireless uplinks:
  Analytical solutions with asymptotic optimality,'' \emph{IEEE Internet of
  Things Journal}, vol.~6, no.~2, pp. 3885--3898, April 2019.

\bibitem{Ayoub2018}
W.~{Ayoub} \emph{et~al.}, ``Internet of mobile things: Overview of {LoRaWAN,
  DASH7, and NB-IoT in LPWANs} standards and supported mobility,'' \emph{IEEE
  Communications Surveys Tutorials}, pp. 1--1, 2018.

\bibitem{Andrews2011}
J.~G. {Andrews}, F.~{Baccelli}, and R.~K. {Ganti}, ``A tractable approach to
  coverage and rate in cellular networks,'' \emph{IEEE Transactions on
  Communications}, vol.~59, no.~11, pp. 3122--3134, November 2011.

\bibitem{Elsawy_tutorial}
H.~{ElSawy} \emph{et~al.}, ``Modeling and analysis of cellular networks using
  stochastic geometry: A tutorial,'' \emph{IEEE Communications Surveys
  Tutorials}, vol.~19, no.~1, pp. 167--203, Firstquarter 2017.

\bibitem{Haenggi2012}
M.~Haenggi, \emph{Stochastic Geometry for Wireless Networks}.\hskip 1em plus
  0.5em minus 0.4em\relax New York, NY, USA: Cambridge University Press, 2012.

\bibitem{Zhong2017}
Y.~{Zhong}, T.~Q.~S. {Quek}, and X.~{Ge}, ``Heterogeneous cellular networks
  with spatio-temporal traffic: Delay analysis and scheduling,'' \emph{IEEE
  Journal on Selected Areas in Communications}, vol.~35, no.~6, pp. 1373--1386,
  June 2017.

\bibitem{Gharbieh2018}
M.~{Gharbieh} \emph{et~al.}, ``Spatiotemporal model for uplink {IoT} traffic:
  Scheduling and random access paradox,'' \emph{IEEE Transactions on Wireless
  Communications}, vol.~17, no.~12, pp. 8357--8372, Dec 2018.

\bibitem{Yang2019}
H.~H. {Yang} and T.~Q.~S. {Quek}, ``Spatiotemporal analysis for {SINR} coverage
  in small cell networks,'' \emph{IEEE Transactions on Communications}, pp.
  1--1, 2019.

\bibitem{Chisci2019}
G.~{Chisci} \emph{et~al.}, ``Uncoordinated massive wireless networks:
  Spatiotemporal models and multiaccess strategies,'' \emph{IEEE/ACM
  Transactions on Networking}, 2019.

\bibitem{YangAoI2019}
Y.~H. {Yang} \emph{et~al.}, ``Locally adaptive scheduling policy for optimizing
  information freshness in wireless networks,'' in \emph{2019 IEEE Global
  Communications Conference (GLOBECOM)}, Dec 2019, pp. 1--6.

\bibitem{Hu2018}
Y.~{Hu}, Y.~{Zhong}, and W.~{Zhang}, ``Age of information in poisson
  networks,'' in \emph{2018 10th International Conference on Wireless
  Communications and Signal Processing (WCSP)}, Oct 2018, pp. 1--6.

\bibitem{ElSawy2014}
H.~{ElSawy} and E.~{Hossain}, ``On stochastic geometry modeling of cellular
  uplink transmission with truncated channel inversion power control,''
  \emph{IEEE Transactions on Wireless Communications}, vol.~13, no.~8, pp.
  4454--4469, Aug 2014.

\bibitem{Huang2015}
L.~{Huang} and E.~{Modiano}, ``Optimizing age-of-information in a multi-class
  queueing system,'' in \emph{2015 IEEE International Symposium on Information
  Theory (ISIT)}, June 2015, pp. 1681--1685.

\bibitem{Haenggi_meta}
M.~{Haenggi}, ``The meta distribution of the sir in poisson bipolar and
  cellular networks,'' \emph{IEEE Transactions on Wireless Communications},
  vol.~15, no.~4, pp. 2577--2589, April 2016.

\bibitem{Haenggi_meta2}
Y.~{Wang}, M.~{Haenggi}, and Z.~{Tan}, ``The meta distribution of the sir for
  cellular networks with power control,'' \emph{IEEE Transactions on
  Communications}, vol.~66, no.~4, pp. 1745--1757, April 2018.

\bibitem{ElSawy_meta}
H.~{ElSawy} and M.~{Alouini}, ``On the meta distribution of coverage
  probability in uplink cellular networks,'' \emph{IEEE Communications
  Letters}, vol.~21, no.~7, pp. 1625--1628, July 2017.

\bibitem{Loynes1962}
R.~M. Loynes, ``The stability of a queue with non-independent inter-arrival and
  service times,'' \emph{Mathematical Proceedings of the Cambridge
  Philosophical Society}, vol.~58, no.~3, p. 497–520, 1962.

\bibitem{Alfa2015}
A.~S.~Alfa, \emph{Applied discrete-time queues, second edition}.\hskip 1em plus
  0.5em minus 0.4em\relax Springer-New York USA, 01 2015.

\bibitem{Zhou2016}
Y.~{Zhou} and W.~{Zhuang}, ``Performance analysis of cooperative communication
  in decentralized wireless networks with unsaturated traffic,'' \emph{IEEE
  Transactions on Wireless Communications}, vol.~15, no.~5, pp. 3518--3530, May
  2016.

\bibitem{Singh2015}
S.~{Singh}, X.~{Zhang}, and J.~G. {Andrews}, ``Joint rate and sinr coverage
  analysis for decoupled uplink-downlink biased cell associations in hetnets,''
  \emph{IEEE Transactions on Wireless Communications}, vol.~14, no.~10, pp.
  5360--5373, Oct 2015.

\end{thebibliography}

	\thispagestyle{empty}
	\appendices

\section{Proof of Lemma 1}\label{se:Appendix_A}

The $b$-th moment of the transmission success probability can be derived from eq.(\ref{eq:PS}) as
\begin{equation}
M_b = \mathbb{E}^{!}_{r_i,P_i,r_o} \Big[ \prod_{r_i \in \Phi_o} \Big(\frac{\bar{\chi}}{1 + \frac{\theta P_ir_o^{\eta(1-\epsilon)}}{\rho r_i^{\eta}}} + 1-\bar{\chi} \Big)^b \Big].
\end{equation}
Moreover, the uplink transmission power $P_i$ of the $i$-th device is a random variable due to the employed fractional path-loss power control. This imposes correlation between the transmission powers of different devices due to the Voronoi sizes correlation  of their serving \acp{BS}. Similar to \cite{Singh2015, ElSawy_meta}, the transmission power correlations are ignored for mathematical tractability. The interfering devices are approximated with the interference seen from an inhomogeneous \ac{PPP} $\tilde{\mathrm{\Phi}} \in \mathbb{R}^2$ with the following intensity function 
$\lambda(x) = \lambda(1-e^{-\pi\lambda||x||^2})$. Following \cite[Theorem 1]{ElSawy_meta} via applying the mapping and displacement theorem on $\tilde{\mathrm{\Phi}}$, the approximated expressions for the moments are derived, where the activity profile depicts the aggregate probability of having a packet within the queue awaiting transmission, which is $1-\chi$, where $\chi$ is the aggregate idle probability.
	
\end{document}